\documentclass[dvipsnames]{article}
\usepackage[utf8]{inputenc}
\usepackage[T1]{fontenc}
\usepackage{adjustbox}
\usepackage{amssymb}
\usepackage{amsthm}
\usepackage{amsmath}
\usepackage{authblk}
\usepackage{blkarray}
\usepackage{booktabs}
\usepackage{bm}
\usepackage{calc}
\usepackage{caption}
\usepackage{cases}
\usepackage{centernot}
\usepackage{dcolumn}
\usepackage{float}
\usepackage[margin=1in]{geometry}
\usepackage{graphicx}
\usepackage{listings}
\usepackage{makecell}
\usepackage{mathtools}
\usepackage[round]{natbib}
\usepackage{pifont}

\usepackage{setspace}
\usepackage{soul}
\usepackage{subcaption}
\usepackage{tabularx}
\usepackage{tgpagella}
\usepackage{tikz}
    \usetikzlibrary{positioning}
\usepackage{tikzsymbols}
\usepackage{tkz-graph}
\usepackage{verbatim}
\usepackage{xcolor}
\usepackage{xfrac}
\usepackage[style=english]{csquotes}
    \MakeOuterQuote{"}
\newtheorem{theorem}{Theorem}
\newtheorem{corollary}{Corollary}
\newtheorem{lemma}{Lemma}

\title{How urban scaling and resource distribution shape social welfare and migration dynamics}
\author{Bryce Morsky\\ \texttt{bmorsky@fsu.edu}}
\affil{Department of Mathematics, Florida State University, Tallahassee, FL, USA}
\date{\today}

\begin{document}

\maketitle

\begin{abstract}
        Many outputs of cities scale in universal ways, including infrastructure, crime, and economic activity. Through a mathematical model, this study investigates the interplay between such scaling laws in human organization and governmental allocations of resources, focusing on impacts to migration patterns and social welfare. We find that if superlinear scaling resources of cities --- such as economic and social activity --- are the primary drivers of city dwellers' utility, then cities tend to converge to similar sizes and social welfare through migration. In contrast, if sublinear scaling resources, such as infrastructure, primarily impact utility, then migration tends to lead to megacities and inequity between large and small cities. These findings have implications for policymakers, economists, and political scientists addressing the challenges of equitable and efficient resource allocation.
\end{abstract}

\section*{Introduction}

Cities exhibit distinct scaling patterns. As they grow, various factors such as infrastructure, individual human needs, and social dynamics scale in non-uniform ways as urbanization reshapes infrastructure demands and economic and social opportunities. This phenomenon, known as urban scaling, has been extensively studied, revealing that different aspects of urban life follow different scaling laws \citep{bettencourt07,bettencourt08,bettencourt13,schlapfer14,west18}. 
Mathematically, this scaling behaviour is captured by the relationship between a quantity $Y$ (e.g,\ infrastructure, economic activity, or social outcomes) and the population size of the city $N$, namely 
\begin{equation}
    Y = Y_0 N^\beta. \label{eq:scaling}
\end{equation}
Here $Y_0$ is a constant baseline and $\beta$ is a scaling exponent that governs how $Y$ changes with city size.

City scaling laws can be classified into three general categories based on the value of the exponent $\beta$, whether it is less than, greater than, or equal to one \citep{bettencourt07}. Sublinear scaling occurs when $\beta<1$. Material infrastructure including roads, utilities, and public services, typically scales sublinearly with population size. Thus, as cities grow, they require proportionally less infrastructure per capita, derived from economies of scale and efficiencies from higher population density. Linear scaling  occurs for $\beta=1$. Some aspects of human needs, such as the provision of basic services (e.g., healthcare, education, and food), tend to scale linearly with population size. For every additional person, a city needs to provide roughly the same amount of these services, resulting in a direct one-to-one relationship. Finally, superlinear scaling ($\beta>1$) is associated with social and economic activities, such as innovation, economic output, crime, and social interactions. Larger cities, due to their density and interconnectedness, generate disproportionately more of these social outputs per capita compared to smaller cities.

These scaling patterns create complex challenges for the fair and efficient allocation of resources, especially between large urban centers and smaller cities or rural areas. For instance, larger cities may have growing demands for resources that cater to superlinear activities such as innovation and economic growth, while smaller cities may require relatively more resources per capita to maintain essential infrastructure. These divergent needs can be highly salient factors in electoral and legislative systems, complicating legislative decisions on resource allocation mechanisms that aim to balance fairness and efficiency across regions of varying sizes. Moreover, since electoral systems and allocation rules may not have been designed to reflect these scaling laws, they may exacerbate inefficiencies and inequities.

The challenges posed by scaling laws are compounded by other factors. Resource allocation by legislative bodies --- whether infrastructure funding, emergency aid (such as COVID-19 or hurricane relief), or community development block grants --- is shaped by social, economic, and political considerations, which may be at odds. In infrastructure funding, for example, a key goal in allotment of resources can be distributive justice \citep{pereira17}, yet such goals may not align with other metrics of fairness or with efficiency. Efficiency based allocations may even exacerbate inequalities. In Turkey, for example, there is a focus on funding already wealthy regions in part due to efficiencies in allocations to these regions \citep{luca15}. In contrast, in the European Union, distributions have been found to be generally aligned with political goals of equity, though some inefficiencies and inequalities persist due to biased allocations favouring in-party voting jurisdictions \citep{dellmuth12}. In practice, such political considerations can heavily affect allocations and lead to inefficiencies and inequities.

Political and electoral influences on resource allocation have been well documented across different democracies \citep{atlas95,finan21,jacques21} with legislators' party affiliations shown to play a significant role in inequitable allocations \citep{larcinese06,sole08,berry10,brollo12,fouirnaies15}. When making allocations, legislators may strategically target districts based on whether they are a party stronghold or a swing district, though this observation is not universal \citep{larcinese13}. In some cases, such as Korea, there is a U-shaped relationship between funding and party affiliation: the president supports jurisdictional strongholds of both his or her own party as well as the opposition party \citep{horiuchi08}. Given the highly varying political incentives and contexts of different, whether a party should target core or swing voters can highly depend on the political and economic situation \citep{dynes20}.

In addition to political and strategic influences, the structure of the electoral system can significantly affect allocations. How legislators are elected --- whether by majoritarian or proportional systems --- can distort representation, thereby influencing allocations across regions. For instance, a higher proportion of representatives to population resulted in higher aid per capita in COVID-19 relief \citep{clemens21}. Beyond any formal faithlessness to representation, informal social networks can further shape allocations by fostering cooperation and coordination among legislators, which can in turn facilitate inequitable resource allocation \citep{jiang20}.

While such political and electoral topics have been well studied, the implications of scaling laws for political economy and democratic decision-making remain underexplored --- particularly in relation to different electoral and legislative systems and with respect to equitable resource allocation. This paper investigates the implications of these scaling laws on the equitable distribution of resources in democratic societies, and evaluates distribution mechanisms in light of these laws. By integrating these scaling laws into utility-based models, we assess their implications for trade-offs between efficiency and equity, as well as their effects on migration patterns and shifting preferences of city dwellers. These insights provide interpretations of migration patterns and can inform legislators addressing the diverse needs of cities and towns within a shared political and economic system.

\section*{Utility, social welfare, and resource allocation}

To model how city dwellers derive utility from the quantity of resources allocated to their cities, we employ utility functions that are functions of the ratios of resources received to those needed, as determined by urban scaling laws. This assumes that the requisite government funding to address certain factors scales as those factors scale for cities. Thus, the utility of city $i$, $U_i$, is a function of $X_{i,j}/Y_{i,j}$, where $X_{i,j}$ represents the resources allocated to a quantity $Y_{i,j} = Y_{0,j} (p_iN)^{\beta_j}$ for city $i$. $N$ is a fixed parameter, the total population across $n$ cities under the government's control, and $p_i \in [0,1]$ is the proportion of this population in city $i$. How $p_i$ may change through migration will be discussed in the next section.

We assume that utility increases with resource allocation, but we also assume that it exhibits diminishing returns. An example such a utility function for individuals in city $i$ is:
    \begin{equation}
        U_i = \sum_{j=1}^m \alpha_j \ln\!\left(1+\frac{X_{i,j}}{Y_{0,j} (p_iN)^{\beta_j}}\right). \label{eq:utility}
    \end{equation}
Here, there are $m$ distinct resources being allocated. These may be funds to build and maintain road networks, water supply, or other infrastructure, in the case that $\beta_j<1$, or funds to manage crime or disease, in the case that $\beta_j>1$. The weights $\alpha_j \geq 0$, with $\sum_{j=1}^m \alpha_j = 1$, reflect the relative impacts of quantity $j$ on an individual's utility. $\alpha_j$ need not be fixed, but could reflect preferences that change as cities grow or shrink. This case will be explored in the next section. Also note that this utility function assumes homogeneity among individuals in a particular city and that cities each have the same functions for each resource.

Given this utility function, we assess the overall social welfare of the total population across all cities using a variety of social welfare functions, which each reflect distinct goals of resource distribution and utility optimization. These welfare functions allow us to evaluate resource divisions given different goals, and have direct policy implications as has been seen in urban planning debates about equity vs.\ efficiency in public goods allocations \citep{glaeser08}. For example, utilitarian social welfare (weighted by population size) maximizes average welfare across all cities:
    \begin{equation}
        W_U = \sum_{i=1}^n p_iU_i. \label{eq:average_welfare}
    \end{equation}
Optimizing this welfare may be the goal of a benevolent central government or stationary bandit that wishes to maximize tax revenue \citep{olson93}. However, maximizing such a welfare function may result in unequal distributions, as it does not inherently address fairness or equity between cities of different sizes. In contrast, a city-weighted average, $W_C = \sum_{i=1}^n U_i/n$, gives equal weight to each city. Going further, Rawlsian social welfare, $W_R = \min(U_1,\ldots U_n)$, focuses on equity, ensuring that the city with the least utility receives priority in resource allocation, thereby addressing the needs of the worst-off city. It is optimized by allocations that equalize utilities across cities. However, since we will explore migrations driven by utility differences, we will only employ Rawlsian social welfare as a benchmark to evaluate other allocation methods.

A balance between the goals of fairness and maximizing tax revenues can be evaluated with the Nash social welfare function:
\begin{equation}
    W_N = \prod_{i=1}^n U_i^{p_i}, \label{eq:Nash_welfare},
\end{equation}
which frequently arises in bargaining settings \citep{nash50,binmore86}. $W_N$ balances efficiency and equity by penalizing low utility in small cities, while still favouring larger populations through the exponent $p_i$. Together, these welfare functions provide a spectrum of allocation philosophies, ranging from pure efficiency (utilitarian) to extreme equity (Rawlsian) and compromise (Nash) in between. Next we consider governments that allocate according to utilitarian or comprise-based principles.

Since allocations to cities can be determined a variety of ways influenced by factors like population size and economic output, it is important to model how the population size of a city maps to the resources allotted it. In a system where resources are allocated directly in proportion to city size --- as in a simple majority voting system --- the allocation is given by $X_{i,j} = \gamma_j p_i$, where $\gamma_j>0$ is a constant representing the total budget for resource $j$. Examples include direct democracy without negotiation, random assemblies \citep{gastil13}, or more generally a legislature governed by a principle of proportionality \citep{brighouse10}. Highly centralized or authoritarian systems can also result in proportional allocations.

Alternatively, resources can be distributed to maximize Nash social welfare. With respect to Nash bargaining, it could arise from a variety of processes and legislative structures. For example, parliamentary systems with proportional representation often lead to coalition governments, which then need to negotiate policies such as resource allocations. Though policies are bargained, the bargaining power of each party may be based on party strength or willingness to fight \citep{harsanyi56}, which is plausibly proportional to the population size. Divided government through bicameral legislatures or a presidential system can also result in this type of bargained solution when there is strong veto power. These examples are far from exhaustive: intergovernmental bargaining between different levels of government, legislative bargaining in parliamentary committees, and consensus democracies are other systems that could also result in a Nash solution. Majoritarian systems, authoritarian, and direct democracies, on the other hand, frequently lead to less bargaining and thus for these scenarios the Nash solution is less relevant.

Up until this point, we have assumed that cities' resources are entirely determined centrally. However, in addition to allocations from a central government, cities can of course self-fund through the economic activity they generate. If such revenue is proportional to the city's size, then implementing a local allocation can be modelled by $X_{i,j} = \tilde{X}_{i,j} + \eta_j p_i N$, where $\tilde{X}_{i,j}$ is the central government's allocation and $\eta_j \geq 0$ is a self-funding parameter for a city. However, this is equivalent to proportional allocations with $\eta_jN$ subsumed into $\gamma_j$. Such endogeneity of resources can also be factored into an allocation rule that aims to optimize Nash social welfare. On the other hand, if the local revenues vary nonlinearly with respect to city size, then such endogeneity could play a more determinative role.

In the next sections, we show how allocation rules and scaling impact migration patterns and how they affect social welfare. And, we explore how city dwellers preferences may evolve in response to these changes.

\section*{Intercity migration}

Here we explore the dynamics of city sizes by assuming that populations are not static but evolve in response to resource allocations. Specifically, the proportions of the population in city $i$, $p_i$, evolve in response to differences between cities' utilities. We model such dynamics using a replicator equation, which is ubiquitous in evolutionary game theory and economics \citep{taylor78,sigmund86,sandholm20,cressman14} and has been previously applied to urban systems, such as road networks \citep{ahmad23}. Its form here is:
    \begin{equation}
        \dot{p}_i = p_i\sum_{j \neq i} p_jr_{i,j}(U_i-U_j). \label{eq:replicator}
    \end{equation}
Individuals thus migrate to cities that have relatively higher utilities for their inhabitants. However, the rate of migration is modulated by the distance between cities, represented by the parameters $r_{i,j} = r_{j,i} > 0$. $r_{i,j}$ measures the closeness or nearness of cities $i$ and $j$, and is smaller the greater the distance between them. For example, $r_{i,j}$ can be defined as the square of the Euclidean distance between cities. Thus, though migration occurs towards cities that have higher utility, the rate of migration is reduced for cities further away. This intuition that distance plays a role in migration has been modeled previously in gravity law migration models \citep{zipf46,karemera00,barbosa18}. This modification to the replicator equation has been explored previously in a different context, homophilic imitation \citep{morsky17}, where it was shown that such modifications generally do not affect the system in the long-run. However, they do have impacts on transitional periods, potentially slowing convergence to equilibria. As we will see, they have a similar effect here.

A key assumption of the above model is that it assumes that (negotiated) allocations are immediately updated and their effects realized as populations change. In reality, there may be a lag in the central government adjusting allocations: elections are periodic and negotiations take time. And even once allocations are adjusted, it may take time for city dwellers to realize their benefits. It takes time to build new infrastructure for example. However, such delays do not affect the stability of equilibria. See Appendix \ref{app:theorems} for the mathematical proofs. As such, though the results presented here assume rapid reallocation of resources in response to evolving populations, they hold true more generally.

We can summarize the analytical findings --- presented in detail in Appendix \ref{app:theorems} --- succinctly as follows. If the weighted average of the scaling exponents is superlinear, then there is a stable equitable equilibrium in which cities are of the same size and individuals receive the same utility. On the other hand, if the weighted average of the scaling exponents is sublinear, then there a single megacity that contains the entire population is a stable equilibrium.

\begin{figure}[htb!]
\centering
\includegraphics[width=0.95\columnwidth]{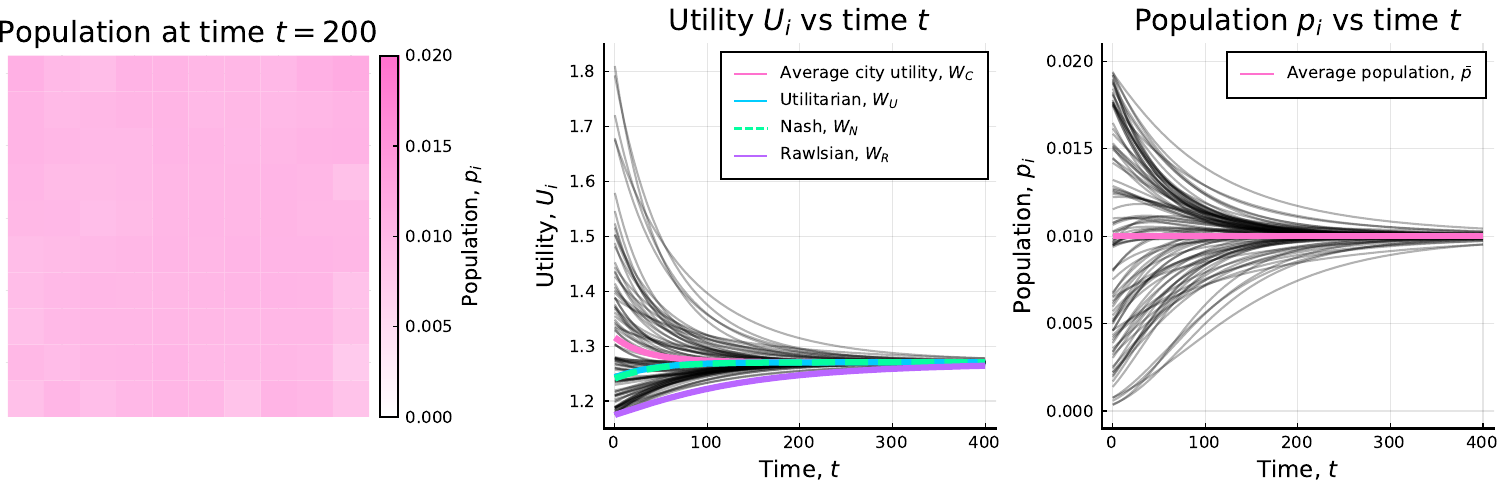}
\caption{Results for scaling exponents $\beta_j=1.1,1.2,1.3$ with $\alpha_j=1/3$. In the long-run, cities converge on the same size.}
\label{fig:superlinear}
\end{figure}

\begin{figure}[htb!]
\centering
\includegraphics[width=0.95\columnwidth]{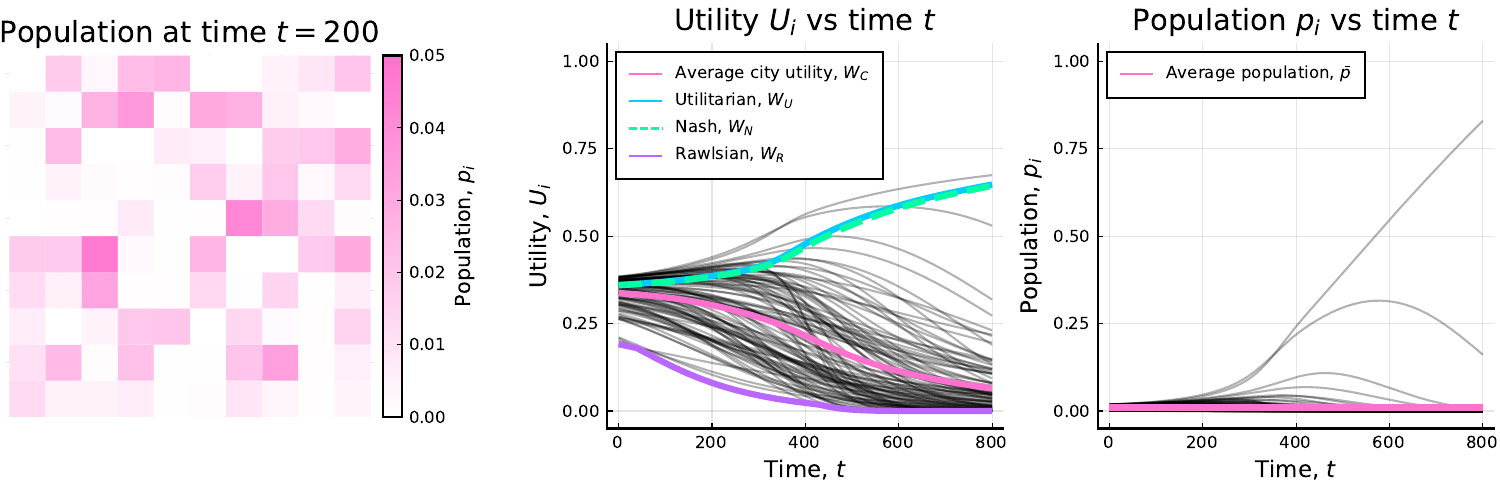}
\caption{Results for scaling exponents $\beta_j=0.7,0.9,1.3$ with $\alpha_j=1/3$. In the long-run, a single megacity dominates. However, multiple cities of varying sizes can coexist during transitional periods.}
\label{fig:sublinear}
\end{figure}

To illustrate the migration dynamics and evolving city utilities, we simulated the system. Figures \ref{fig:superlinear} and \ref{fig:sublinear} depict typical outcomes for proportional allocation of resources when the scaling factors are superlinear and sublinear, respectively. Appendix \ref{app:methods} depicts similar results for allocations that optimize Nash social welfare. These plots include heatmaps representing the spatial distribution of the population at an intermediate time to showcase transitional states of the system. In the superlinear scaling scenario, the system rapidly converges to a homogeneous state where all cities have the same size and provide the same utility. When the scaling exponents are sublinear, however, the dynamics lead to fragmentation: some cities become significantly larger while others shrink or remain small. Due to how distances between cities slows migration, many large cities can coexist in a transitional period. However, slowly, one city will attract all individuals leading ultimately to the formation of a single stable megacity.

\section*{Shifting preferences}

The preceding analyses assumed that preferences for resources --- represented by the weights $\alpha_j$ --- are identical across city size, space, and time. In reality, they may be functions of city size, or vary geographically due to regional differences in culture and economics. City dwellers may value or prioritize different resources, which may be a response to the utility they generate. Shifting preferences could switch the weighted average scaling exponent from superlinear to sublinear, for example, drastically altering the trajectory of a city's population. To study how preferences may shift over time and in turn affect the qualitative outcome of migration dynamics, we apply the framework of adaptive dynamics \citep{diekmann04,brannstrom13}. Consider a two city scenario under proportional allocation of two different resources, and assume that preferences change more slowly than migration dynamics. Without loss of generality, let $\alpha_i$ be the preference of the first resource in city $i$ and assume that $\beta_1 > \beta_2$. We aim to understand the conditions under which $\alpha_i$ for either city makes small incremental shifts over time. 

As is standard in the adaptive dynamics framework, we consider a rare individual in city $i$  with a different preference $\tilde{\alpha}_i \neq \alpha_i$ than the majority of individuals in that city. Since this individual is rare, we assume that the equilibria of the migration dynamics are solely determined by the majority population. We then calculate the difference between the utilities of the rare individual and the majority evaluated at this equilibrium:
\begin{equation}
    S(\tilde{\alpha}_i,\alpha_i) = U_i(\tilde{\alpha}_i,\alpha_i) - U_i(\alpha_i,\alpha_i).
\end{equation}
Taking the derivative of this equation and evaluating it at $\tilde{\alpha}_i=\alpha$ gives us the \emph{selection gradient} $S'(\alpha_i)$, whose sign determines whether or not preferences shift. If $S'(\alpha_i) > 0$, $\alpha_i$ increases, and if $S'(\alpha_i) < 0$, it decreases. Essentially, we assume that preferences shift if they marginally improve the utility of city dwellers.

Under a proportional allocation rule, preferences depend on the term:
\begin{equation}
    \Theta = \frac{\gamma_1Y_{0,2}N^{\beta_2}}{\gamma_2Y_{0,1}N^{\beta_1}},
\end{equation}
which is the relative ratios of allocations received to those needed for each resource (see Appendix \ref{app:theorems} for details). If $\Theta \geq 1$, preferences will always shift to increasingly weight the relatively higher scaling exponent (i.e.\ $\alpha_i$ will increase). Such a scenario can occur if $\gamma_i = Y_{0,i}N^{\beta_i}$. However, if $\gamma_1 < Y_{0,1}N^{\beta_1}$ while $\gamma_2 \geq Y_{0,2}N^{\beta_2}$, then the selection gradient could be negative leading to preferences shifting toward the second (relatively sublinear) resource.

\begin{figure}[htb!]
\captionsetup[subfigure]{justification=centering}
    \centering
    \begin{subfigure}[]{0.425\columnwidth}
        \caption{Bifurcation diagram, $\Theta=0.8$}\label{fig:bif_Theta_0.8}
        \includegraphics[width=\textwidth]{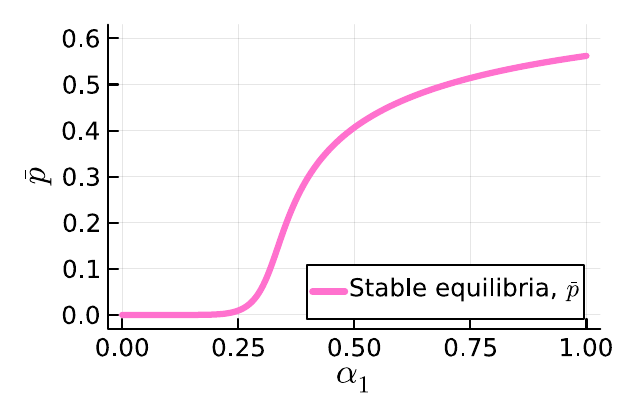}
    \end{subfigure}
    \begin{subfigure}[]{0.425\columnwidth}
        \caption{Invasibility diagram, $\Theta=0.8$}\label{fig:PIP_Theta_0.8}
        \includegraphics[width=\textwidth]{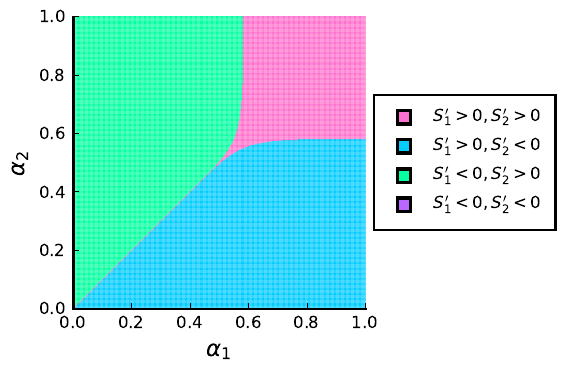}
    \end{subfigure} \\
    \begin{subfigure}[]{0.425\columnwidth}
        \caption{Bifurcation diagram, $\Theta=0.6$}\label{fig:bif_Theta_0.6}
        \includegraphics[width=\textwidth]{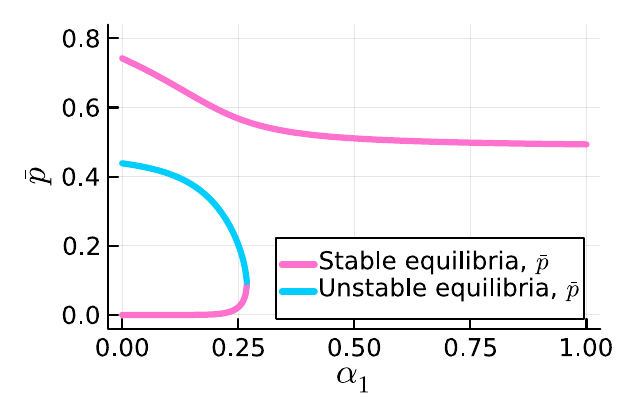}
    \end{subfigure}
    \begin{subfigure}[]{0.425\columnwidth}
        \caption{Invasibility diagram, $\Theta=0.6$}\label{fig:PIP_Theta_0.6}
        \includegraphics[width=\textwidth]{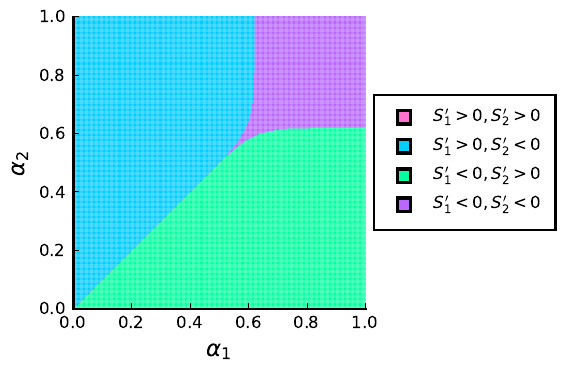}
    \end{subfigure}
    \caption{Bifurcation and invasibility diagrams. $\Theta=0.8$ in the top row and $\Theta=0.6$ in the bottom. The first column depicts bifurcation diagrams detailing the stable and unstable equilibria $\bar{p} \in (0,1)$ for varying values of the preference for the superlinear resource for city $1$, $\alpha_1$, with $\alpha_2$ fixed at $0.7$. To determine the equilibria at which preferences shift in the invasibility plots, the population is initialized at the equilibrium reached from the initial condition $p(0)=0.5$. $\beta_1=1.4$, $\beta_2=0.7$, $\gamma_1 = \Theta Y_{0,1}N^{\beta_1}$, and $\gamma_2 = Y_{0,2}N^{\beta_2}$. The remaining parameter values can be found in Appendix \ref{app:methods}.\\ Panel (a) shows that migrations lead to city $2$ growing into a megacity for sufficiently low $\alpha_1$. However, beyond a threshold, the cities can coexist with the size of city $1$ growing as $\alpha_1$ increases. Panel (c) shows a bistable scenario for sufficiently low $\alpha_1$: either city $2$ is a megacity, or both cities can coexist with city $1$ larger than city $2$. This bistability is lost beyond a threshold $\alpha_1$. When $\Theta=0.8$, preference shifting leads to either one megacity or two equally sized cities as depicted in Panel (b). When $\Theta = 0.6$, however, there is only one attractor of preference evolution, two cities of equal size and utility (Panel (d)).}
    \label{fig:PIP}
\end{figure}

To visualize whether or not preferences shift and their effects on equilibria, we plot bifurcation and invasibility diagrams in Figure \ref{fig:PIP}. The latter diagram depicts the direction of preference shifts by labelling regions according to the signs of the selection gradients for each city. If $\Theta$ is sufficiently large, we observe three attractors of shifting preferences. In one, city $1$ develops into a megacity while city. The preferences of dwellers of the larger city shift to weighting the sublinear resource more than superlinear one. And, preferences of dwellers in the smaller city shift to prefer the superlinear resource. However, in so doing, this drives the growth of a megacity. The third attractor of preference evolution is where dwellers of both cities weight the superlinear resource maximally ($\alpha_1=\alpha_2=1$), which in turn leads to the cities being of equal size. In contrast to these results, for a sufficiently low $\Theta$ (e.g.,\ $\Theta=0.6$), there is a sole attractor of preference evolution in which both cities coexist at equal sizes and there is an intermediate weighting of the resources. Though coexistence of the cities is achieved in this case, overall utility is lower.

\section*{Discussion}

Here we have explored the implications of city scaling laws --- which describe how city outputs scale with population size --- to social welfare and migration dynamics. We examined different allocations rules and how city dwellers preferences may shift in response to funding levels, and we showed that our results are robust to basic assumptions about utility functions and the timing re-allocations of resources. We uncovered that such scaling laws can either exacerbate inequality, driving the growth of megacities, or drive dispersal of the population geographically, conditional on the scaling exponents. Urban bias in policy allocation has historically influenced both democratic and authoritarian regimes. In authoritarian contexts, urban bias can undermine regime stability if policies disproportionately favour urban centres at the expense of rural areas \citep{wallace13}, while in democracies inequitable resource distribution could drive political polarization and undermine democratic norms. Since allocation mechanisms may not adequately account for cities' needs reflected by scaling laws, it is important for policy makers to be aware of their impacts. Budget distribution frameworks could be developed that align with the needs of different cities, particularly through the adoption of welfare functions that consider both efficiency and equity. More broadly, this approach highlights the importance of considering population scale and heterogeneity across cities in the design of democratic institutions with an aim to foster a more stable, inclusive, and representative political landscape.

The results of this study highlight the complex interplay between urban scaling laws, resource distribution, and allocation mechanisms. However, several important avenues for future research remain. For one, due to the theoretical nature of this study, empirical validation is essential. By analyzing historical data on city growth, resource distribution, and political representation, the model parameters and features could be refined. Another critical assumption is within city homogeneity: individuals have the same preferences and reap the same benefits from living within the same city. In reality, cities differ widely in terms of their internal demographics, economic structure, and social priorities. Thus, within city heterogeneity is a key model extension to explore.

While the current model focuses on resource allocation rules, future work could explore legislators' strategic considerations in designing and implementing these rules. Voter targeting, political clientism, vote buying, and city lobbying \citep{cox86,wantchekon03,khemani15,goldstein17} --- which can undermine efficient and equitable resource transfers and allocations --- are such possible strategical extensions to the model. Additionally, infrastructure or other resources that are partially or wholly shared between regions may produce a free-rider effect \citep{besley03}. Roads, in particular, are often public goods that may benefit neighbouring regions. These factors may support the coexistence of smaller cities near larger ones, unlike the current model.

Finally, our results have implications for the urban growth equation \citep{bettencourt07}:
\begin{equation}
    \frac{dN}{dt} = \frac{Y_0}{E}N^\beta - \frac{R}{E}N,
\end{equation}
where $R$ and $E$ are the per-capita quantities required to maintain the population and grow it, respectively. In the case of superlinear scaling, growth is unbounded. As discussed in \cite{bettencourt07}, such growth could lead to economic and social collapse. However, our findings show that superlinear scaling drives migration to redistribute populations, balancing them across many cities. Thus, so long as space and infrastructure allow such dispersal, our results suggest that migration may mitigate such risk of stagnation and collapse.

\subsection*{Code and data availability}
Code to numerically solve the models and plot the results is available at github.com/bmorsky/scaling.

\bibliography{scaling}
\bibliographystyle{apalike}

\appendix

\section{Theorems} \label{app:theorems}

Here we provide mathematical details concerning the stability of equilibria of the Equation \ref{eq:replicator} and the selection gradient. Additionally, we consider an extension of the model to a system of delay differential equation, where there is a delay $\tau>0$ in the realization of reallocations. Thus, allocations are of the form $X_{i,j}(p_i(t-\tau))$ and utility is:
    \begin{equation}
        U_i(t) = U(p_i(t),p_i(t-\tau)) = \sum_{j=1}^m \alpha_j \ln\!\left(1 + \frac{X_{i,j}(p_i(t-\tau))}{Y_{0,j} (p_i(t)N)^{\beta_j}}\right). \label{eqn:utility_lag}
    \end{equation}
Theorems \ref{thm:superlinear} and \ref{thm:sublinear} apply to the base model and Corollaries \ref{thm:superlinear_lag} and \ref{thm:sublinear_lag} apply to the model with delays.

\begin{theorem} \label{thm:superlinear}
    Let $\bar{p}_i = \bar{p}_j = \bar{p} \: \forall i,j$. If the weighted average scaling exponent is superlinear, then the equilibrium is stable. If it is sublinear, then the equililbrium is unstable.
\end{theorem}
    
\begin{proof}
    First we define the utility $U_i = \sum_{k=1}^m \alpha_k U_{i,k}$ and note the following:
    \begin{equation}
        \frac{\partial}{\partial p_i}(U_i-U_j) = \sum_{k=1}^m \alpha_k \frac{\partial}{\partial p_i}U_
        {i,k}\!\left(\frac{X_{i,k}}{Y_{0,k}(p_iN)^{\beta_k}}\right) = \sum_{k=1}^m \alpha_k \frac{U_{i,k}'}{Y_{0,k}(p_iN)^{\beta_k}}\!\left(X_{i,k}' - \beta_k\frac{X_{i,k}}{p_i}\right),
    \end{equation}
    where $U_{i,k}'$ and $X_{i,k}'$ are the derivatives of $U_{i,k}$ and $X_{i,k}$, respectively. The equilibrium $\bar{p}_i=\bar{p}_j=\bar{p} \: \forall i,j \implies U_i=U_j \: \forall i,j$. Then, linearizing about this equilibrium, the entries of the Jacobian matrix $\mathbf{J} = \{j_{i,j}\}$ with $i\neq j$ are:
    \begin{align}
        j_{i,i} &= \sum_{k=1}^m \bar{p}r_{i,k}\underbrace{(U_i-U_k)}_{0} + \bar{p}^2 r_{i,k}\frac{\partial U_i}{\partial p_i} = \bar{p}^2\bar{r}_i\sum_{l=1}^m \alpha_l \frac{U_{i,l}'}{Y_{0,l}(\bar{p}N)^{\beta_l}}\!\left(X_{i,l}' - \beta_l\frac{X_{i,l}}{\bar{p}}\right), \\
        j_{i,j} &= \bar{p}r_{i,j}\underbrace{(U_i-U_j)}_{0} - \bar{p}^2r_{i,j}\frac{\partial U_j}{\partial p_j} = - \bar{p}^2r_{i,j} \sum_{l=1}^m \alpha_l \frac{U_{j,l}'}{Y_{0,l}(\bar{p}N)^{\beta_l}}\!\left(X_{j,l}' - \beta_l\frac{X_{j,l}}{\bar{p}}\right),
    \end{align}
    where $\bar{r}_i = \sum_{j=1}^n r_{i,j}$. Note that though there may be a different utility function for each resource, we are assuming homogeneity among cities. Thus, $U_{i,l}'=U_{j,l}'$, $X_{i,l}=X_{j,l}$, and $X_{i,l}'=X_{j,l}' \: \forall i,j$. The Jacobian is thus:
    \begin{equation}
        \mathbf{J} = \underbrace{\left(\bar{p}^2\sum_{l=1}^m \alpha_l \frac{U_{i,l}'}{Y_{0,l}(\bar{p}N)^{\beta_l}} \left(X_{i,l}' - \beta_l\frac{X_{i,l}}{\bar{p}}\right)\right)}_{\Phi}
        \underbrace{\begin{bmatrix}
        \bar{r}_1 & -r_{1,2} & \cdots & -r_{1,n} \\
        -r_{1,2} & \bar{r}_2 & \cdots & -r_{2,n} \\
        \vdots & \cdots & \ddots & -r_{n-1,n} \\
        -r_{1,n} & -r_{2,n} & \cdots & \bar{r}_n
        \end{bmatrix}}_{\mathbf{M}},
    \end{equation}
    Note that the matrix $\mathbf{M} = \{m_{ij}\}$ is symmetric, has positive diagonal entries, and is zero-line (and row) sum. Particularly, it is positive semi-definite. Its eigenvalues are thus positive but for a zero eigenvalue corresponding to the constraint $\sum_{i=1}^n p_i=1$. Thus, $\mathbf{J}$ has all negative eigenvalues if $\Phi<0$.
    
    Recall that $U'_{il}>0$ by the assumption that utility is increasing with respect to resources allocated to the city. Thus, in the case where utility functions are the same across resource ($U_i=U_{ij}=U_{ik}$), $\Phi<0$ simplifies to:
    \begin{equation}
        \sum_{l=1}^m \kappa_l\left(X_{i,l}' - \beta_l\frac{X_{i,l}}{\bar{p}}\right) < 0, \label{eq:thm1_condition}
    \end{equation}
    where $\kappa_l = \alpha_l/(Y_{0,l}(\bar{p}N)^{\beta_l})$. In the scenario of proportional distribution ($X_{i,l} = \gamma_l\bar{p}$), this condition simplifies further to $\sum_{l=1}^m \tilde{\kappa}_l \beta_l > 1$ where $\tilde{\kappa}_l = \gamma_l\kappa_l/(\sum_{l=1}^m \gamma_l\kappa_l)$ (i.e\ the weighted average is superlinear). Further, since there are diminishing returns to social welfare under the Nash social welfare function, the optimal allocation rule is a concave function. Thus, $X_{i,l}'<X_{i,l}/\bar{p}$. And therefore, this condition will also hold for legislatures that allocate by the Nash rule.
\end{proof}

\begin{corollary} \label{thm:superlinear_lag}
    Let $\bar{p}_i = \bar{p}_j \: \forall i, j$. Theorem \ref{thm:superlinear} holds for the corresponding system of delay differential equations where utility is defined by Equation \ref{eqn:utility_lag}.
\end{corollary}
\begin{proof}
    Linearizing with respect to $p(t)$ and $p(t-\tau)$, we have the following characteristic equation:
    \begin{equation}
        \det(\lambda\mathbf{I} - (\Phi_1e^{-\lambda_i\tau} - \Phi_2)\mathbf{M}) = 0,
    \end{equation}
    with
    \begin{align}
        \Phi_1 &= \bar{p}^2\sum_{l=1}^m \alpha_l \frac{U_{i,l}'}{Y_{0,l}(\bar{p}N)^{\beta_l}}X_{i,l}' > 0, \\
        \Phi_2 &= \bar{p}^2\sum_{l=1}^m \alpha_l \frac{U_{i,l}'}{Y_{0,l}(\bar{p}N)^{\beta_l}}\beta_l\frac{X_{i,l}}{\bar{p}} > 0.
    \end{align}
    Note that $\Phi_1 < \Phi_2$ is the condition for stability in Theorem \ref{thm:superlinear}. Since $\mathbf{M}$ is a real symmetric matrix, we may diagonalize it with invertible matrix $\mathbf{P}$ such that $\mathbf{M}=\mathbf{P}\mathbf{D}\mathbf{P}^{-1}$. We may then rewrite the characteristic equation:
    \begin{equation}
        \det(\lambda\mathbf{I} - (\Phi_1e^{-\lambda_i\tau} - \Phi_2)\mathbf{M}) = \det(\lambda \mathbf{I} - (\Phi_1e^{-\lambda_i\tau} - \Phi_2)\mathbf{D}) = 0.
    \end{equation}
    We thus have $n$ linear equations of the form $\lambda_i = (\Phi_1e^{-\lambda_i\tau} - \Phi_2)\mu_i$, where $\mu_i \geq 0$ are the diagonal entries of $\mathbf{D}$. Letting $\lambda_i = a + bi$ have positive real part and taking the modulus of both sides of the equation $\lambda_i/\mu_i + \Phi_2 = \Phi_1e^{-\lambda_i\tau}$, we have:
    \begin{equation}
        \Phi_2 < \sqrt{(a/\mu_i+\Phi_2)^2+(b/\mu_i)^2} = \Phi_1e^{-a\tau} < \Phi_1.
    \end{equation}
    However, $\Phi_2 < \Phi_1$ is a contradiction.
\end{proof}

\begin{theorem} \label{thm:sublinear}
    Let $\bar{p}_i = 1$ for some $i$, and thus $\bar{p}_j=0 \: \forall i \neq j$. If the weighted average scaling exponent is sublinear, then each $\bar{p}_i=1$ is a stable equilibrium.  If it is sublinear, then the equilibria are all unstable.
\end{theorem}
\begin{proof}
    We require that $U_i$ is an increasing function with respect to $p_i$ near $\bar{p}_i=1$. Taking the derivative we then have the condition:
    \begin{equation}
        \sum_{l=1}^m \kappa_l(X_{i,l}' - \beta_l X_{i,l}) > 0.
    \end{equation}
    where $\kappa_l = \alpha_l/(Y_{0,l}N^{\beta_l})$. In the scenario of proportional distribution ($X_{i,l} = \gamma_l$), this condition simplifies further to $\sum_{l=1}^m \tilde{\kappa}_l \beta_l < 1$ where $\tilde{\kappa}_l = \gamma_l\kappa_l/(\sum_{l=1}^m \gamma_l\kappa_l)$ (i.e\ the weighted average is sublinear).
\end{proof}

\begin{corollary} \label{thm:sublinear_lag}
    Let $\bar{p}_i = 1$ for some $i$, and thus $\bar{p}_j=0 \: \forall i \neq j$. Then, Theorem \ref{thm:sublinear} holds for the corresponding system of delay differential equations where utility is defined by Equation \ref{eqn:utility_lag}.
\end{corollary}
\begin{proof}
    Linearizing with respect to $p(t-\tau)$ and evaluating at the equilibrium conditions, the coefficient of the term $e^{-\lambda\tau}$ in the linearization is zero. Thus, the characteristic equation is identical to the non-delay equation, and the stability conditions remain unchanged from that of Theorem \ref{thm:sublinear}.
\end{proof}

\begin{lemma}
    Let allocations be proportional and $\beta_1>\beta_2$. Then, $\Theta \geq 1$ is a sufficient condition for $S'(\alpha_i) > 0$.
\end{lemma}
\begin{proof}
    The selection gradient is:
\begin{equation}
    S'(\alpha_i) = \frac{\partial S(\tilde{\alpha}_i,\alpha_i)}{\partial \tilde{\alpha}_i} \bigg|_{\tilde{\alpha}_i=\alpha_i} = \ln\!\left(1 + \frac{\gamma_1p^{1-\beta_1}}{Y_{0,1}N^\beta_1} \right) - \ln\!\left(1 + \frac{\gamma_2p^{1-\beta_2}}{Y_{0,2}N^\beta_2} \right).
\end{equation}
$S'(\alpha_i)$ is positive if:
\begin{equation}
    \ln\!\left(1 + \frac{\gamma_1p^{1-\beta_1}}{Y_{0,1}N^\beta_1} \right) > \ln\!\left(1 + \frac{\gamma_2p^{1-\beta_2}}{Y_{0,2}N^\beta_2} \right)
    \implies \frac{\gamma_1p^{1-\beta_1}}{Y_{0,1}N^\beta_1} > \frac{\gamma_2 p^{1-\beta_2}}{Y_{0,2}N^\beta_2}
    \implies \Theta p^{\beta_2-\beta_1} > 1.
\end{equation}
Since $\beta_2-\beta_1 < 0 \implies p^{\beta_2-\beta_1} > 1$, $\Theta \geq 1$ is a sufficient condition for $S'(\alpha_i) > 0$.
\end{proof}

\section{Simulation methods and Nash social welfare allocations} \label{app:methods}

Here we detail the simulation methods and present supplementary results for a scenario in which allocations are determined by a legislature maximizing the Nash social welfare function.

All simulations were implemented in Julia and use the DifferentialEquations package \citep{rackauckas17}. The setup consists of $10^2$ cities arranged in a square grid. Three resources are considered with weights $\alpha_j=1/3$ and constant baselines $Y_{0,j}=1$ for all $j$. Utility is defined by Equation \ref{eq:utility}, and the total population size is normalized to $N=1$. $r_{ij}$ is the square of the Euclidean distance between cities $i$ and $j$ on the grid.

The Interior Point OPTimizer (Ipopt) \citep{wachter06} and the JuMP packages \citep{lubin23} are used to compute the allocations that maximize Nash social welfare assuming a total budget $B=1$. Note that the budget can be either divided into fixed sub-budgets for each allocation, $B_j$, or allocations can be drawn from a single common budget. In the former scenario, $\sum_{i=1}^n X_{i,j} = B_j$ and $\sum_{j=1}^m = B$. In the latter, $\sum_{i=1}^n\sum_{j=1}^m X_{i,j} = B$. Here we depict results for the latter scenario, though both results are similar (and both are implemented in the code). Additionally, we also assume that there is no local funding of allocations: they are entirely determined by the central government's allocations.

\begin{figure}[ht!]
\centering
\includegraphics[width=\columnwidth]{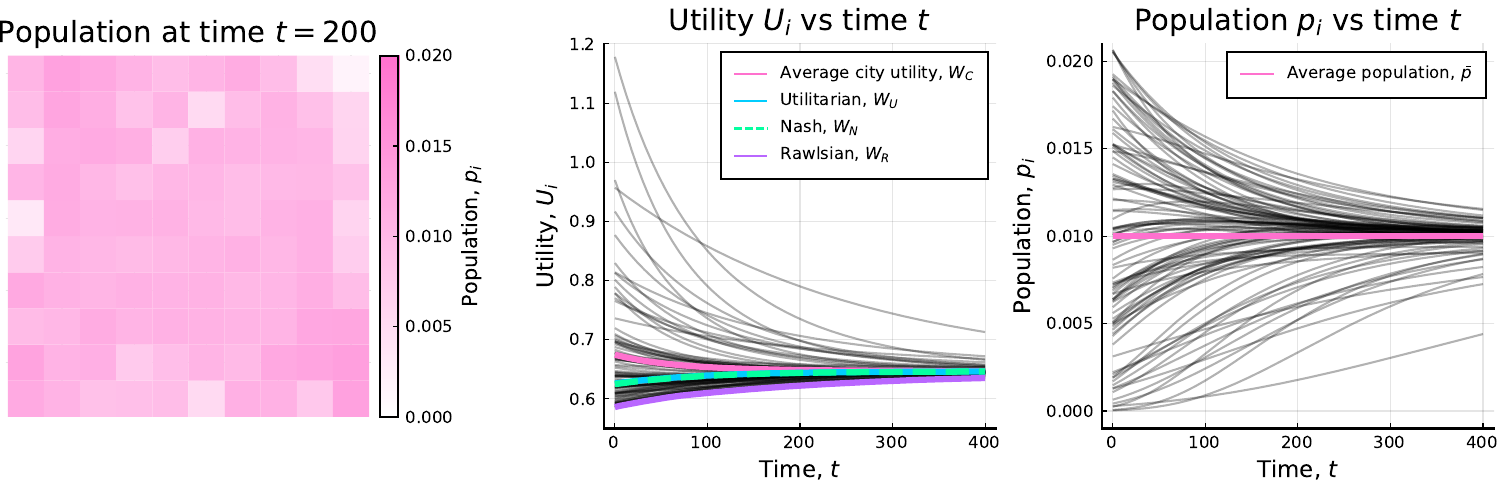}
\caption{Results for scaling exponents $\beta_j=1.1,1.2,1.3$ with $\alpha_j=1/3$. In the long-run, cities converge on the same size.}
\label{fig:Nash_superlinear}
\end{figure}

\begin{figure}[ht!]
\centering
\includegraphics[width=\columnwidth]{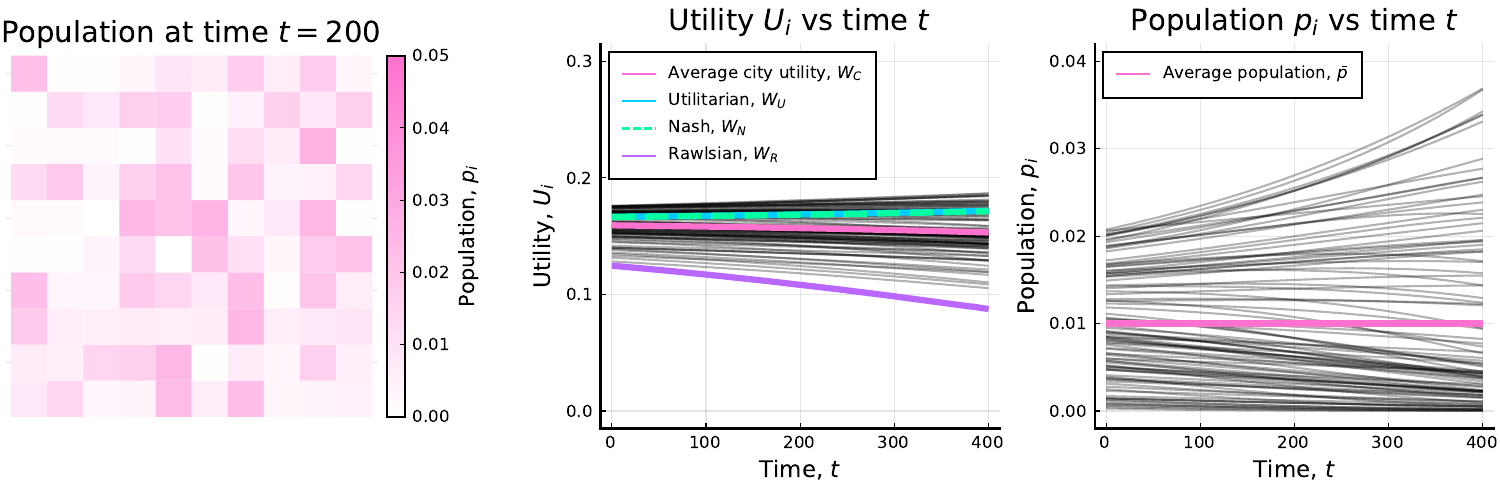}
\caption{Results for scaling exponents $\beta_j=0.7,0.9,1.3$ with $\alpha_j=1/3$. In the long-run, a single megacity dominates. However, multiple cities of varying sizes can coexist during transitional periods.}
\label{fig:Nash_sublinear}
\end{figure}

Figures \ref{fig:Nash_superlinear} and \ref{fig:Nash_sublinear} illustrate the heatmaps and times series for this case. Note that these results are qualitatively similar to those in the main text (Figures \ref{fig:superlinear} and \ref{fig:sublinear}). Superlinear scaling leads to cities converging on the same size and utility. Although, as shown in the proofs in the previous section, sublinear scaling ultimately results in a single megacity. Figure \ref{fig:Nash_sublinear} shows that this convergence can be very slow, resulting in a prolonged transitional period with cities of varying sizes. 

\end{document}